\documentclass[submission,copyright,creativecommons]{eptcs}
\providecommand\event{QPL 2017}
\providecommand\titlerunning{Two Roads to Classicality}
\providecommand\authorrunning{B. Coecke, J. Selby \& S. Tull}

\title{Two Roads to Classicality}

\author{Bob Coecke\thanks{Supported by AFOSR grant Algorithmic and Logical Aspects when Composing Meanings.}
    \institute{University of Oxford}
    \email{bob.coecke@cs.ox.ac.uk}
  \and John Selby\thanks{Supported by the EPSRC through the Controlled Quantum Dynamics Centre for Doctoral Training.}
    \institute{University of Oxford \& Imperial College London}
    \email{john.selby08@imperial.ac.uk}
  \and Sean Tull  \thanks{Supported by EPSRC Studentship OUCL/2014/SET.}
    \institute{University of Oxford}
        \email{sean.tull@cs.ox.ac.uk}
}

\usepackage{rotating}
\usepackage{floatpag}
\rotfloatpagestyle{empty}

\usepackage{hyperref}

\usepackage{amsmath,amsthm,amssymb}
\usepackage{xspace,enumerate,color,epsfig}
\usepackage{graphicx}
\graphicspath{{.}{./figures/}}
\usepackage{marvosym}

\usepackage{tikzfig}
\usepackage{stmaryrd}
\usepackage{docmute}
\usepackage{keycommand}

\newkeycommand{\pointmap}[style=point][1]{\,\tikz{\node[style=\commandkey{style}] (x) at (0,-0.3) {$#1$};\node [style=none] (3) at (0, 1.0) {};\node [style=none] (3) at (0, -1.0) {};\draw (x)--(0,0.7);}\,}

\newkeycommand{\typedkpoint}[style=kpoint,edgestyle=][2]{%
\,\begin{tikzpicture}
  \begin{pgfonlayer}{nodelayer}
    \node [style=none] (0) at (0, 1) {};
    \node [style=\commandkey{style}] (1) at (0, -0.75) {$#2$};
    \node [style=right label] (2) at (0.25, 0.5) {$#1$};
  \end{pgfonlayer}
  \begin{pgfonlayer}{edgelayer}
    \draw [\commandkey{edgestyle}] (1) to (0.center);
  \end{pgfonlayer}
\end{tikzpicture}\,}

\newkeycommand{\typedkpointdag}[style=kpoint adjoint,edgestyle=][2]{%
\,\begin{tikzpicture}
  \begin{pgfonlayer}{nodelayer}
    \node [style=none] (0) at (0, -1) {};
    \node [style=\commandkey{style}] (1) at (0, 0.75) {$#2$};
    \node [style=right label] (2) at (0.25, -0.5) {$#1$};
  \end{pgfonlayer}
  \begin{pgfonlayer}{edgelayer}
    \draw [\commandkey{edgestyle}] (0.center) to (1);
  \end{pgfonlayer}
\end{tikzpicture}\,}

\newcommand{\bistateadj}[1]{\,\begin{tikzpicture}[yshift=-3mm]
  \begin{pgfonlayer}{nodelayer}
    \node [style=kpointadj, minimum width=1 cm, inner sep=2pt] (0) at (0, 0.5) {$\psi$};
    \node [style=none] (1) at (-0.75, 0.25) {};
    \node [style=none] (2) at (0.75, 0.25) {};
    \node [style=none] (3) at (-0.75, -0.5) {};
    \node [style=none] (4) at (0.75, -0.5) {};
  \end{pgfonlayer}
  \begin{pgfonlayer}{edgelayer}
    \draw (1.center) to (3.center);
    \draw (2.center) to (4.center);
  \end{pgfonlayer}
\end{tikzpicture}\,}

\newkeycommand{\pointketbra}[style1=copoint,style2=point][2]{\,%
\begin{tikzpicture}
  \begin{pgfonlayer}{nodelayer}
    \node [style=none] (0) at (0, -1.5) {};
    \node [style=\commandkey{style1}] (1) at (0, -0.75) {$#1$};
    \node [style=\commandkey{style2}] (2) at (0, 0.75) {$#2$};
    \node [style=none] (3) at (0, 1.5) {};
  \end{pgfonlayer}
  \begin{pgfonlayer}{edgelayer}
    \draw (0.center) to (1);
    \draw (2) to (3.center);
  \end{pgfonlayer}
\end{tikzpicture}\,}

\newkeycommand{\twocopointketbra}[style1=copoint,style2=copoint,style3=point][3]{\,%
\begin{tikzpicture}
  \begin{pgfonlayer}{nodelayer}
    \node [style=none] (0) at (-0.75, -1.5) {};
    \node [style=none] (1) at (0.75, -1.5) {};
    \node [style=\commandkey{style1}] (2) at (-0.75, -0.75) {$#1$};
    \node [style=\commandkey{style2}] (3) at (0.75, -0.75) {$#2$};
    \node [style=\commandkey{style3}] (4) at (0, 0.75) {$#3$};
    \node [style=none] (5) at (0, 1.5) {};
  \end{pgfonlayer}
  \begin{pgfonlayer}{edgelayer}
    \draw (0.center) to (2);
    \draw (1.center) to (3);
    \draw (4) to (5.center);
  \end{pgfonlayer}
\end{tikzpicture}\,}

\newkeycommand{\twopointketbra}[style1=copoint,style2=point,style3=point][3]{\,%
\begin{tikzpicture}
  \begin{pgfonlayer}{nodelayer}
    \node [style=none] (0) at (0, -1.5) {};
    \node [style=none] (1) at (-0.75, 1.5) {};
    \node [style=\commandkey{style1}] (2) at (0, -0.75) {$#1$};
    \node [style=\commandkey{style2}] (3) at (-0.75, 0.75) {$#2$};
    \node [style=\commandkey{style3}] (4) at (0.75, 0.75) {$#3$};
    \node [style=none] (5) at (0.75, 1.5) {};
  \end{pgfonlayer}
  \begin{pgfonlayer}{edgelayer}
    \draw (0.center) to (2);
    \draw (1.center) to (3);
    \draw (4) to (5.center);
  \end{pgfonlayer}
\end{tikzpicture}\,}

\newkeycommand{\pointbraket}[style1=point,style2=copoint][2]{\,%
\begin{tikzpicture}
  \begin{pgfonlayer}{nodelayer}
    \node [style=\commandkey{style1}] (0) at (0, -0.5) {$#1$};
    \node [style=\commandkey{style2}] (1) at (0, 0.5) {$#2$};
  \end{pgfonlayer}
  \begin{pgfonlayer}{edgelayer}
    \draw (0) to (1);
  \end{pgfonlayer}
\end{tikzpicture}\,}

\newkeycommand{\kpointbraket}[style1=kpoint,style2=kpoint adjoint][2]{\,%
\begin{tikzpicture}
  \begin{pgfonlayer}{nodelayer}
    \node [style=\commandkey{style1}] (0) at (0, -0.75) {$#1$};
    \node [style=\commandkey{style2}] (1) at (0, 0.75) {$#2$};
  \end{pgfonlayer}
  \begin{pgfonlayer}{edgelayer}
    \draw (0) to (1);
  \end{pgfonlayer}
\end{tikzpicture}\,}

\newkeycommand{\kpointmap}[style1=kpoint,style2=map][2]{\,%
\begin{tikzpicture}
  \begin{pgfonlayer}{nodelayer}
    \node [style=\commandkey{style1}] (0) at (0, -1.1) {$#1$};
    \node [style=\commandkey{style2}] (1) at (0, 0.2) {$#2$};
    \node [style=none] (2) at (0, 1.5) {};
  \end{pgfonlayer}
  \begin{pgfonlayer}{edgelayer}
    \draw (0) to (1);
    \draw (1) to (2.center);
  \end{pgfonlayer}
\end{tikzpicture}%
\,}

\newkeycommand{\sandwichtwo}[style1=point,style2=map,style3=map,style4=copoint][4]{\,%
\begin{tikzpicture}
  \begin{pgfonlayer}{nodelayer}
    \node [style=\commandkey{style2}] (0) at (0, -0.75) {$#2$};
    \node [style=\commandkey{style3}] (1) at (0, 0.75) {$#3$};
    \node [style=\commandkey{style1}] (2) at (0, -2) {$#1$};
    \node [style=\commandkey{style4}] (3) at (0, 2) {$#4$};
  \end{pgfonlayer}
  \begin{pgfonlayer}{edgelayer}
    \draw (2) to (0);
    \draw (0) to (1);
    \draw (1) to (3);
  \end{pgfonlayer}
\end{tikzpicture}\,}

\newkeycommand{\longbraket}[style1=point,style2=copoint][2]{\,%
\begin{tikzpicture}
  \begin{pgfonlayer}{nodelayer}
    \node [style=\commandkey{style1}] (0) at (0, -2) {$#1$};
    \node [style=\commandkey{style2}] (1) at (0, 2) {$#2$};
  \end{pgfonlayer}
  \begin{pgfonlayer}{edgelayer}
    \draw (0) to (1);
  \end{pgfonlayer}
\end{tikzpicture}\,}

\newkeycommand{\onb}[style=point]{\ensuremath{\left\{\tikz{\node[style=\commandkey{style}] (x) at (0,-0.3) {$j$};\draw (x)--(0,0.7);}\right\}}\xspace}

\newkeycommand{\copointmap}[style=copoint][1]{\,\tikz{\node[style=\commandkey{style}] (x) at (0,0.3) {$#1$};\draw (x)--(0,-0.7);}\,}

\newcommand{\catC}{\ensuremath{\mathcal{C}}\xspace}

\newcommand{\catD}{\ensuremath{\mathcal{D}}\xspace}

\newcommand{\catE}{\ensuremath{\mathcal{E}}\xspace}

\tikzstyle{cdiag}=[matrix of math nodes, row sep=3em, column sep=3em, text height=1.5ex, text depth=0.25ex,inner sep=0.5em]
\tikzstyle{arrow above}=[transform canvas={yshift=0.5ex}]
\tikzstyle{arrow below}=[transform canvas={yshift=-0.5ex}]

\newcommand{\sclr}{s}

\newcommand{\QuantPure}{\cat{Quant}_{\text{pure}}}
\newcommand{\QuantMixed}{\cat{Quant}} 
\newcommand{\CStar}{\cat{CStar}}

\newcommand{\leakminC}{\catC_L}

\newcommand{\ClassProb}{\cat{Class}}
\newcommand{\Rel}{\cat{Rel}}
\newcommand{\FRel}{\cat{FRel}}

\newcommand{\boolB}{\mathbb{B}}

\newcommand{\Modalp}{\cat{Modal}_p}

\newcommand{\genU}[1]{{#1}_{\langle U \rangle}}
\newcommand{\catCU}{\genU{\catC}}

\newcommand{\FHilb}{\cat{FHilb}}
\newcommand{\CMon}{\cat{CMon}}
\newcommand{\MatR}{\Mat_R}
\newcommand{\Mat}{\cat{Mat}}

\let\olddagger\dagger
\renewcommand{\dagger}{\ensuremath{\olddagger}\xspace}

\newcommand{\pure}{\text{pure}}
\newcommand{\op}{\text{op}}
\newcommand{\id}[1]{\ensuremath{\mathrm{id}_{#1}}}
\newcommand{\cat}[1]{\ensuremath{\mathbf{#1}}\xspace}
\newcommand{\coproj}{\kappa}

\newcommand{\Split}{\mathsf{Split}}

\newcommand{\Splitorcausal}{\Split^{(\hspace{-.5pt}\ttinyground\hspace{-1.5pt})}\hspace{-1pt}}
\newcommand{\Splitcausal}{\Split^{\ttinyground}\hspace{-1pt}}
\newcommand{\Splitdagcausal}{\Split_{\hspace{0.6pt}\dagger}^{\ttinyground}\hspace{-1pt}}
\newcommand{\hilbH}{\mathcal{H}} 
\newcommand{\hilbK}{\mathcal{K}} 

\newcommand{\CPM}{\mathsf{CPM}\xspace}
\newcommand{\CPs}{\mathsf{CP}^*\xspace}

\newcommand{\pdeff}[1]{\bar{#1}}
\newcommand{\pds}{a}
\newcommand{\pda}{a}
\newcommand{\pdb}{b}
\newcommand{\pdc}{c}

\usepackage{amsmath,amssymb,stmaryrd,xspace,enumerate}
\usepackage{tikz}
\usetikzlibrary{matrix}
\usetikzlibrary{decorations.markings}
\usetikzlibrary{shapes.geometric}
\usetikzlibrary{arrows,positioning}
\usetikzlibrary{calc}
\usetikzlibrary{matrix,shapes}
\usepackage[all,cmtip]{xy} 
\usepackage{wasysym}
\usepackage{tikz-cd}
\usepackage{mathtools}
\usepackage{multirow}
\usepackage{amsthm}
\usepackage{bbm}
\usepackage{hyperref}
\usepackage{graphicx}

\usepackage{amsmath,amssymb,stmaryrd,xspace,xypic,enumerate}
\usepackage{tikz,tikzfig}

\makeatletter
\def\th@plain{%
  \thm@notefont{}
  \itshape 
}
\def\th@definition{%
  \thm@notefont{}
  \normalfont 
}
\makeatother

\tikzset{
    >=stealth',
    punkt/.style={
           rectangle,
           rounded corners,
           draw=black, 
           text width=6.5em,
           minimum height=2em,
           text centered},
    pil/.style={
           ->,
           shorten <=2pt,
           shorten >=2pt,},
     incl/.style={
           left hook->,
           shorten <=2pt,
           shorten >=2pt,},
     incl2/.style={
           right hook->,
           shorten <=2pt,
           shorten >=2pt,},
}

\newcommand{\discard}[1]{\ensuremath{\ \tinyground_{#1}}}

\tikzstyle{dot}=[circle, draw=black, fill=black!25, inner sep=.4ex]
\tikzstyle{black_dot}=[dot, fill=black!50]
\tikzstyle{white_dot}=[dot, fill=white]

\newif\ifvflip\pgfkeys{/tikz/vflip/.is if=vflip}
\newif\ifhflip\pgfkeys{/tikz/hflip/.is if=hflip}
\newif\ifhvflip\pgfkeys{/tikz/hvflip/.is if=hvflip}
\newlength\morphismheight
\newlength\wedgewidth
\tikzset{width/.initial=1mm}
\makeatletter
\tikzstyle{morphism}=[font=\small,morphismshape]
\pgfdeclareshape{morphismshape}
{
    \savedanchor\centerpoint
    {
        \pgf@x=0pt
        \pgf@y=0pt
    }
    \anchor{center}{\centerpoint}
    \anchorborder{\centerpoint}
    \saveddimen\overallwidth
    {
        \pgfkeysgetvalue{/tikz/width}{\minwidth}
        \pgf@x=\wd\pgfnodeparttextbox
        \ifdim\pgf@x<\minwidth
            \pgf@x=\minwidth
        \fi
    }
    \savedanchor{\upperrightcorner}
    {
        \pgf@y=.5\ht\pgfnodeparttextbox
        \advance\pgf@y by -.5\dp\pgfnodeparttextbox
        \pgf@x=.5\wd\pgfnodeparttextbox
    }
    \anchor{north}
    {
        \pgf@x=0pt
        \pgf@y=0.5\morphismheight
    }
    \anchor{north east}
    {
        \pgf@x=\overallwidth
        \multiply \pgf@x by 2
        \divide \pgf@x by 5
        \pgf@y=0.5\morphismheight
    }
    \anchor{east}
    {
        \pgf@x=\overallwidth
        \divide \pgf@x by 2
        \advance \pgf@x by 5pt
        \pgf@y=0pt
    }
    \anchor{west}
    {
        \pgf@x=-\overallwidth
        \divide \pgf@x by 2
        \advance \pgf@x by -5pt
        \pgf@y=0pt
    }
    \anchor{north west}
    {
        \pgf@x=-\overallwidth
        \multiply \pgf@x by 2
        \divide \pgf@x by 5
        \pgf@y=0.5\morphismheight
    }
    \anchor{south east}
    {
        \pgf@x=\overallwidth
        \multiply \pgf@x by 2
        \divide \pgf@x by 5
        \pgf@y=-0.5\morphismheight
    }
    \anchor{south west}
    {
        \pgf@x=-\overallwidth
        \multiply \pgf@x by 2
        \divide \pgf@x by 5
        \pgf@y=-0.5\morphismheight
    }
    \anchor{south}
    {
        \pgf@x=0pt
        \pgf@y=-0.5\morphismheight
    }
    \anchor{text}
    {
        \upperrightcorner
        \pgf@x=-\pgf@x
        \pgf@y=-\pgf@y
    }
    \backgroundpath
    {
    \begin{scope}
        \pgfkeysgetvalue{/tikz/fill}{\morphismfill}
        \pgfsetstrokecolor{black}
        \pgfsetlinewidth{.7pt}
        \begin{scope}
        \pgfsetstrokecolor{black}
        \pgfsetfillcolor{white}
        \pgfsetlinewidth{.7pt}
                \ifhflip
                    \pgftransformyscale{-1}
                \fi
                \ifvflip
                    \pgftransformxscale{-1}
                \fi
                \ifhvflip
                    \pgftransformxscale{-1}
                    \pgftransformyscale{-1}
                \fi
                \pgfpathmoveto{\pgfpoint
                    {-0.5*\overallwidth-5pt}
                    {0.5*\morphismheight}}
                \pgfpathlineto{\pgfpoint
                    {0.5*\overallwidth+5pt}
                    {0.5*\morphismheight}}
                \pgfpathlineto{\pgfpoint
                    {0.5*\overallwidth + \wedgewidth}
                    {-0.5*\morphismheight}}
                \pgfpathlineto{\pgfpoint
                    {-0.5*\overallwidth-5pt}
                    {-0.5*\morphismheight}}
                \pgfpathclose
                \pgfusepath{fill,stroke}
        \end{scope}
    \end{scope}
    }
}
\makeatother

\newenvironment{pic}[1][]
{\begin{aligned}\begin{tikzpicture}[font=\tiny,#1]}
{\end{tikzpicture}\end{aligned}}

\pgfdeclarelayer{foreground}
\pgfdeclarelayer{background}
\pgfdeclarelayer{morphismlayer}
\pgfsetlayers{background,main,morphismlayer,foreground}

\tikzstyle{dot}=[circle, draw=black, fill=black!20, inner sep=.4ex, node on layer=foreground]

\tikzstyle{whitedot}=[circle, draw=black, fill=white, inner sep=.4ex, node on layer=foreground]
\tikzstyle{greydot}=[circle, draw=black, fill=black!20, inner sep=.4ex, node on layer=foreground]
\tikzstyle{darkgreydot}=[circle, draw=black, fill=black!50, inner sep=.4ex, node on layer=foreground]
\tikzstyle{blackdot}=[circle, draw=black, fill=black, inner sep=.4ex, node on layer=foreground]

\tikzstyle{triangle} = [regular polygon, regular polygon sides=3, draw=black, fill=black!20,scale=0.3, node on layer=foreground]

\tikzstyle{whitetriangle}=[triangle, fill=white]
\tikzstyle{greytriangle}=[triangle, fill=black!20]
\tikzstyle{darkgreytriangle}=[triangle, fill=black!50]
\tikzstyle{blacktriangle}=[triangle, fill=black]

\tikzstyle{invertedtriangle} = [triangle,scale=-1]
\tikzstyle{whiteinvertedtriangle}=[invertedtriangle, fill=white]
\tikzstyle{greyinvertedtriangle}=[invertedtriangle, fill=black!20]
\tikzstyle{darkgreyinvertedtriangle}=[invertedtriangle, fill=black!50]
\tikzstyle{blackinvertedtriangle}=[invertedtriangle, fill=black]

\tikzstyle{morphism}=[font=\small,morphismshape]
\tikzstyle{box}=[rectangle,inner sep=.4ex, draw=black, node on layer=foreground]
\tikzstyle{mor}=[rectangle,inner sep=.4ex, draw=black, node on layer=foreground, minimum width = 0.6cm, minimum height = 0.4cm, font = \small]
\tikzstyle{morwide}=[mor, minimum width = 1.2cm]

\newif\ifvflip\pgfkeys{/tikz/vflip/.is if=vflip}
\newif\ifhflip\pgfkeys{/tikz/hflip/.is if=hflip}
\newif\ifhvflip\pgfkeys{/tikz/hvflip/.is if=hvflip}
\tikzset{width/.initial=1mm}

\newlength\minimummorphismwidth
\newlength\stateheight
\newlength\minimumstatewidth
\newlength\connectheight
\tikzset{colour/.initial=white}

\tikzstyle{mixed}=[line width=.7pt]
\tikzstyle{pure}=[line width=.7pt]
\tikzset{diredge/.style={decoration={
  markings,
  mark=at position 0.525 with {\arrow{#1}}},postaction={decorate}}}
\tikzset{
    diredge/.default=>
}

\tikzset{diredgestart/.style={decoration={
  markings,
  mark=at position 4pt with {\arrow{#1}}},postaction={decorate}}}
\tikzset{
    diredgestart/.default=<
}

\tikzset{diredgeend/.style={decoration={
  markings,
  mark=at position 1 with {\arrow{#1}}},postaction={decorate}}}
\tikzset{
    diredgeend/.default=>
}

\makeatletter

\pgfdeclareshape{morphismshape}
{
    \savedanchor\centerpoint
    {
        \pgf@x=0pt
        \pgf@y=0pt
    }
    \anchor{center}{\centerpoint}
    \anchorborder{\centerpoint}
    \saveddimen\overallwidth
    {
        \pgfkeysgetvalue{/tikz/width}{\minwidth}
        \pgf@x=\wd\pgfnodeparttextbox
        \ifdim\pgf@x<\minwidth
            \pgf@x=\minwidth
        \fi
    }
    \savedanchor{\upperrightcorner}
    {
        \pgf@y=.5\ht\pgfnodeparttextbox
        \advance\pgf@y by -.5\dp\pgfnodeparttextbox
        \pgf@x=.5\wd\pgfnodeparttextbox
    }
    \anchor{north}
    {
        \pgf@x=0pt
        \pgf@y=0.5\morphismheight
    }
    \anchor{north east}
    {
        \pgf@x=\overallwidth
        \multiply \pgf@x by 4
        \divide \pgf@x by 5
        \pgf@y=0.5\morphismheight
    }
    \anchor{east}
    {
        \pgf@x=\overallwidth
        \divide \pgf@x by 2
        \advance \pgf@x by 5pt
        \pgf@y=0pt
    }
    \anchor{west}
    {
        \pgf@x=-\overallwidth
        \divide \pgf@x by 2
        \advance \pgf@x by -5pt
        \pgf@y=0pt
    }
    \anchor{north west}
    {
        \pgf@x=-\overallwidth
        \multiply \pgf@x by 4
        \divide \pgf@x by 5
        \pgf@y=0.5\morphismheight
    }
    \anchor{south east}
    {
        \pgf@x=\overallwidth
        \multiply \pgf@x by 4
        \divide \pgf@x by 5
        \pgf@y=-0.5\morphismheight
    }
    \anchor{south west}
    {
        \pgf@x=-\overallwidth
       \multiply \pgf@x by 4
        \divide \pgf@x by 5
        \pgf@y=-0.5\morphismheight
    }
    \anchor{south}
    {
        \pgf@x=0pt
        \pgf@y=-0.5\morphismheight
    }
    \anchor{text}
    {
        \upperrightcorner
        \pgf@x=-\pgf@x
        \pgf@y=-\pgf@y
    }
    \backgroundpath
    {
    \begin{scope}
        \pgfkeysgetvalue{/tikz/fill}{\morphismfill}
        \pgfsetstrokecolor{black}
        \begin{scope}
        \pgfsetstrokecolor{black}
        \pgfsetfillcolor{white}
                \ifhflip
                    \pgftransformyscale{-1}
                \fi
                \ifvflip
                    \pgftransformxscale{-1}
                \fi
                \ifhvflip
                    \pgftransformxscale{-1}
                    \pgftransformyscale{-1}
                \fi
                \pgfpathmoveto{\pgfpoint
                    {-0.5*\overallwidth-5pt}
                    {0.5*\morphismheight}}
                \pgfpathlineto{\pgfpoint
                    {0.5*\overallwidth+5pt}
                    {0.5*\morphismheight}}
                \pgfpathlineto{\pgfpoint
                    {0.5*\overallwidth + \wedgewidth}
                    {-0.5*\morphismheight}}
                \pgfpathlineto{\pgfpoint
                    {-0.5*\overallwidth-5pt}
                    {-0.5*\morphismheight}}
                \pgfpathclose
                \pgfusepath{fill,stroke}
        \end{scope}
    \end{scope}
    }
}

\pgfdeclareshape{groundNEW}
{
    \savedanchor\centerpoint
    {
        \pgf@x=0pt
        \pgf@y=0pt
    }
    \anchor{center}{\centerpoint}
    \anchorborder{\centerpoint}

    \anchor{north}
    {
        \pgf@x=0pt
        \pgf@y=0.16\stateheight
    }
    \anchor{south}
    {
        \pgf@x=0pt
        \pgf@y=0pt
    }
    \saveddimen\overallwidth
    {
        \pgfkeysgetvalue{/pgf/minimum width}{\minwidth}
        \pgf@x=\minimumstatewidth
        \ifdim\pgf@x<\minwidth
            \pgf@x=\minwidth
        \fi
    }
    \backgroundpath
    {
        \begin{pgfonlayer}{foreground}
        \pgfsetstrokecolor{black}
        \pgfsetlinewidth{1.25pt}
        \ifhflip
            \pgftransformyscale{-1}
        \fi
        \pgftransformscale{0.5}
        \pgfpathmoveto{\pgfpoint{-0.5*\overallwidth}{0pt}}
        \pgfpathlineto{\pgfpoint{0.5*\overallwidth}{0pt}}
        \pgfpathmoveto{\pgfpoint{-0.33*\overallwidth}{0.33*\stateheight}}
        \pgfpathlineto{\pgfpoint{0.33*\overallwidth}{0.33*\stateheight}}
        \pgfpathmoveto{\pgfpoint{-0.16*\overallwidth}{0.66*\stateheight}}
        \pgfpathlineto{\pgfpoint{0.16*\overallwidth}{0.66*\stateheight}}
        \pgfpathmoveto{\pgfpoint{-0.02*\overallwidth}{\stateheight}}
        \pgfpathlineto{\pgfpoint{0.02*\overallwidth}{\stateheight}}
        \pgfusepath{stroke}
        \end{pgfonlayer}
    }
}

\pgfdeclareshape{groundNEWsmall}
{
    \savedanchor\centerpoint
    {
        \pgf@x=0pt
        \pgf@y=0pt
    }
    \anchor{center}{\centerpoint}
    \anchorborder{\centerpoint}

    \anchor{north}
    {
        \pgf@x=0pt
        \pgf@y=0.16\stateheight
    }
    \anchor{south}
    {
        \pgf@x=0pt
        \pgf@y=0pt
    }
    \saveddimen\overallwidth
    {
        \pgfkeysgetvalue{/pgf/minimum width}{\minwidth}
        \pgf@x=\minimumstatewidth
        \ifdim\pgf@x<\minwidth
            \pgf@x=\minwidth
        \fi
    }
    \backgroundpath
    {
        \begin{pgfonlayer}{foreground}
        \pgfsetstrokecolor{black}
        \pgfsetlinewidth{0.5pt}
        \ifhflip
            \pgftransformyscale{-1}
        \fi
        \pgftransformscale{0.5}
        \pgfpathmoveto{\pgfpoint{-0.5*\overallwidth}{0pt}}
        \pgfpathlineto{\pgfpoint{0.5*\overallwidth}{0pt}}
        \pgfpathmoveto{\pgfpoint{-0.33*\overallwidth}{0.33*\stateheight}}
        \pgfpathlineto{\pgfpoint{0.33*\overallwidth}{0.33*\stateheight}}
        \pgfpathmoveto{\pgfpoint{-0.16*\overallwidth}{0.66*\stateheight}}
        \pgfpathlineto{\pgfpoint{0.16*\overallwidth}{0.66*\stateheight}}
        \pgfpathmoveto{\pgfpoint{-0.02*\overallwidth}{\stateheight}}
        \pgfpathlineto{\pgfpoint{0.02*\overallwidth}{\stateheight}}
        \pgfusepath{stroke}
        \end{pgfonlayer}
    }
}
\tikzstyle{groundwide}=[groundNEW, minimum width = 1.2cm]

\newcommand{\tinyground}[1][groundNEW]{
\smash{\raisebox{-2pt}{\hspace{-3pt}\ensuremath{\begin{pic}[scale=0.4]
    \node[#1, scale=0.6] (1) at (0,0.4) {};
    \draw [pure] (1.south) to +(0,-.3);
\end{pic}
}\hspace{-1pt}}}}

\newcommand{\ttinyground}[1][groundNEWsmall]{
\smash{\raisebox{-0.5pt}{\hspace{-1pt}\ensuremath{\begin{pic}[scale=0.2]
    \node[#1, xscale=0.25,yscale=0.3] (1) at (0,0.4) {};
    \draw [pure] (1.south) to +(0,-.3);
\end{pic}
}\hspace{-0pt}}}}

\usepackage{makeidx}
\makeindex

\theoremstyle{definition}
\newtheorem{theorem}{Theorem}[section]
\newtheorem{corollary}[theorem]{Corollary}
\newtheorem{lemma}[theorem]{Lemma}
\newtheorem{proposition}[theorem]{Proposition}

\newtheorem{definition}[theorem]{Definition}

\newtheorem{example}[theorem]{Example}
\newtheorem{examples}[theorem]{Examples}
\newtheorem{example*}[theorem]{Example*}
\newtheorem{examples*}[theorem]{Examples*}
\newtheorem{remark}[theorem]{Remark}
\newtheorem{remark*}[theorem]{Remark*}

\usepackage{color}
\def\bR{\begin{color}{red}}
\def\bB{\begin{color}{blue}}
\def\bM{\begin{color}{magenta}}
\def\bC{\begin{color}{cyan}}
\def\bW{\begin{color}{white}}
\def\bBl{\begin{color}{black}}
\def\bG{\begin{color}{green}}
\def\bY{\begin{color}{yellow}}
\def\e{\end{color}\xspace}
\newcommand{\bit}{\begin{itemize}}
\newcommand{\eit}{\end{itemize}\par\noindent}
\newcommand{\ben}{\begin{enumerate}}
\newcommand{\een}{\end{enumerate}\par\noindent}
\newcommand{\beq}{\begin{equation}}
\newcommand{\eeq}{\end{equation}\par\noindent}
\newcommand{\beqa}{\begin{eqnarray*}}
\newcommand{\eeqa}{\end{eqnarray*}\par\noindent}
\newcommand{\beqn}{\begin{eqnarray}}
\newcommand{\eeqn}{\end{eqnarray}\par\noindent}

\usepackage{enumitem}

\newcommand{\seanignore}[1]{} 
\newcommand{\johnignore}[1]{} 

\begin{document}

\maketitle

\begin{abstract}
Mixing and decoherence are both manifestations of classicality within quantum theory, each of which admit a very general category-theoretic construction. We show under which conditions these two `roads to classicality' coincide. This is indeed the case for (finite-dimensional) quantum theory, where each construction yields the category of C*-algebras and completely positive maps. We present counterexamples where the property fails which includes relational and modal theories. Finally, we provide a new interpretation for our category-theoretic generalisation of decoherence in terms of `leaking information'.
\end{abstract}

Physical, computational, and many other theories can very generally be described by (monoidal) categories. Examples include categorical logic \cite{LambekScott}, categorical programming language semantics \cite{asperti1991categories}, and more recently,  categorical quantum mechanics~\cite{AC1}. More specifically, we think of any (monoidal) category as a candidate theory of physical systems (objects) and processes (morphisms).
When viewing morphisms as quantum processes, two universal constructions provide roads to classical physics, allowing one to build new systems to describe classical data, respectively embodying a generalisation of \emph{mixing} and of \emph{decoherence}.

Firstly, one may represent mixing in a category $\catC$ by
means of sum-enrichment. This is justified by the fact that, when $\catC$ is monoidal, endomorphisms of the tensor unit (i.e.~\em scalars\em) allow one to then form weighted mixtures of morphisms, generalising the probabilistic mixtures appearing in information theory. When $\catC$ is sum-enriched, we may apply a universal construction, the \emph{biproduct completion} $\catC^{\oplus}$~\cite{MacLane} to generate classical set-like systems, with the biproduct playing the role of the set-union.

The second road to classicality, decoherence, is given in the quantum formalism by an idempotent, causal operation that sets all off-diagonal entries of a density matrix to $0$. Causality may be discussed in any category $\catC$ coming with suitable `discarding' morphisms, and we generalise decoherence to any causal idempotent in such a category.\seanignore{
} Applying our next universal construction, (a variant of) the \emph{Karoubi Envelope} $\Splitcausal(\catC)$ (splitting of idempotents or \emph{ Cauchy completion}) \cite{borceux1986cauchy} generates systems equipped with such a `decoherence' map which `classicise' their processes.

In this article, we investigate when there is an embedding $\catC^{\oplus} \to \Splitcausal(\catC)$ (Theorem~\ref{thm:SplitHasBiprod}) and when it is an equivalence (Corollary~\ref{cor:IdemBiprodEquiv}), showing that this is the case for quantum theory (Corollary~\ref{cor:quantum_main_result}). In this case one recovers not only classical theory, but also intermediate systems described by C*-algebras, and remarkably C*-algebras only.
In its most concrete form, our main result may be stated as follows: in the category of finite-dimensional C*-algebras, all causal idempotents split.

Correspondences between these different manners of encoding classicality within categorical quantum mechanics were already studied by Heunen, Kissinger and Selinger~\cite{EPTCS171.7}. Our result strengthens and greatly generalises theirs by removing the assumption that the idempotents are self-adjoint. Abstractly, this allows our approach to apply beyond the `dagger compact categories' considered there to arbitrary categories without a dagger, or in fact even monoidal structure. Nonetheless, it is a straightforward corollary that our result respects the monoidal structure when present. While the passage from self-adjoint idempotents to general idempotents might seem minor, it is precisely this relaxation that allows for a clear interpretation.

There was a third road in \cite{EPTCS171.7} sandwiched between the two others, corresponding to the use of `dagger Frobenius algebras' as a generalisation of C*-algebras. This leads to another interesting interpretation of our result: as already hinted at above, a structure much weaker than the full-blown axiomatization suffices to capture all finite dimensional C*-algebras. Moreover this weaker structure has a very clear interpretation as resulting from leakage of information into the environment, see Section~\ref{sec:leaks}. Two extremes are the fully quantum C*-algebras with minimal leakage, and the fully classical C*-algebras with maximal leakage. Hence, from a physical perspective, the additional structure of C*-algebras is merely an artefact of the Hilbert space representation.

This work is related to and draws on several earlier works. In particular, Corollary~\ref{cor:IdemBiprodEquiv} draws on a result of Blume-Kohout et al.~\cite{blume2010information}, and our approach can be seen as a generalisation of that of Selinger and collaborators~\cite{SelingerIdempotent,EPTCS171.7} as discussed in depth in Section~\ref{sec:Daggers}.

\section{Setup}

Physical theories can be described as categories where we think of the objects in a category as (physical) \em systems\em, of the morphisms as \em processes\em, and $\circ$ as sequential composition of these processes\footnote{
Indeed we will use the terms category/morphism/object and theory/process/system interchangeably throughout.
}.
We often call a process an \emph{event} when we think of it as forming a part of a probabilistic process (e.g.~the occurrence of a particular outcome in a quantum measurement).

The categories we consider here will typically come with a chosen object to represent `nothing', denoted $I$. We call morphisms $a \colon I\to A$ \em states\em, $e \colon A\to I$ \em effects\em, and $\sclr \colon I \to I$ \em scalars\em. Admittedly, this terminology is slightly abusive unless we take $I$ to be the tensor unit in a monoidal category $(\catC, \otimes, I)$.  In such a case, where we think of the tensor product $\otimes$ as \em parallel composition \em of processes, we can introduce or remove the `nothing'-object at will via the natural isomorphisms
$\lambda_A \colon I \otimes A \to A$ and $\rho_A \colon A \otimes I \to A$.
\seanignore{
It is only then that states/effects acquire their usual meaning, and scalars behave as expected, in particular, inducing a scalar multiplication (see below).
}
While most of our example categories are monoidal, our results do not require any monoidal structure, though importantly they are compatible with any which is present (see Remark~\ref{rem:MonoidalStructure}.

\seanignore{
however, our main results will not require $I$ to be the monoidal unit, and holds for general categories. But importantly, they are entirely compatible with the monoidal structure should it exist (see Section~\ref{Sec:MonoidalStructure}).
}
\seanignore{
For a complete categorical description of quantum theory -- including how collections of events form more general processes -- see \cite{CKBook}.
}

\begin{example}[Pure quantum events]\label{Ex:pureQuantum}
These can be described by a category $\QuantPure$ where the objects are finite-dimensional Hilbert spaces $\hilbH$ and morphisms are linear maps identified up to a global phase (i.e. $f\sim g \iff f=e^{i\theta} g \text{ for some } \theta$). The monoidal product is the standard Hilbert space tensor product and the monoidal unit is $\mathbb{C}$. States are therefore vectors in $\hilbH$ up to a phase, and effects dual-vectors again up to a phase. Scalars are complex numbers up to a phase, i.e.~positive real numbers. \seanignore{, which we think of as probabilistic weights if $\leq 1$.}
\end{example}

Two important structures are lacking in this theory of pure quantum events, namely our two roads to classicality: \emph{discarding} and \emph{mixing}.

\seanignore{
}
\begin{definition}
A \emph{category with discarding} $(\catC, \discard{})$ is a category coming with a chosen object $I$ and family of morphisms $\discard{A} \colon A \to I$, with $\discard{I} = \id{I}$.
\seanignore{
}
A morphism $f \colon A \to B$ is \emph{causal} when
$\discard{B} \circ f = \discard{A}$.
\end{definition}

The reason why the term causality is justified is that when restricting to causal processes, a (monoidal) theory is non-signalling \cite{Cnonsig}.
Hence compatibility with relativity theory boils down to the requirement that all effects are causal and hence equal to the discarding effect; this notion of causality was first introduced in \cite{Chiri1}. However, in this form causality has a very lucid interpretation:
whether we discard an object before or after applying some process is irrelevant, either way the result is the same --- the object ends up discarded \cite{CKpaperI, CKBook}.
This guarantees that processes outside the direct surroundings of an experiment
can be ignored, and hence is vital to even be able to perform any kind of scientific experiment without, for example, intervention from another galaxy. We now move on to generalised mixing.

\begin{definition}
A category $\catC$ is \emph{semi-additive} when it is enriched in the category $\cat{CMon}$ of commutative monoids. That is, each homset forms a commutative
monoid  $(\catC(A,B), +, 0)$, with $+$ and $0$ preserved by composition. In particular, $\catC$ has a family of \emph{zero morphisms} $0 = 0_{A,B} \colon A \to B$, for all $A, B \in \catC$, satisfying $0 \circ f = 0 = 0 \circ g$ for all morphisms $f, g$.
\seanignore{
\[
\begin{tikzcd}[every label/.append style = {font = \small},column sep = large]
A \rar{0= 0_{A,B}} & B
\end{tikzcd}
\]
for $A, B \in \catC$, satisfying
\[
0 \circ f = 0\quad\qquad\mbox{and}\quad\qquad g \circ 0 = 0
\]
for all morphisms $f, g$.
}
We often write $\sum^n_{i=1} f_i$ for $f_1 + \dots + f_n$.
\end{definition}

\begin{remark}
We use the term mixing here as in quantum theory it is precisely this semi-additive structure that allows one to discuss probabilistic mixtures of processes. More generally, in any monoidal category where the scalars $s \colon I \to I$ can be interpreted as probabilities, we may discuss probabilistic weightings of processes by setting $\sclr \cdot f = \lambda_B \circ (s \otimes f) \circ {\lambda_A}^{-1} \colon A \to B$
, and hence probabilistic mixtures $\sum_i \sclr_i \cdot f_i$.
\seanignore{
of processes as:
\[
\sum_i \sclr_i \cdot f_i
\]
}
If $\sum_i \sclr_i=1$ then $\{\sclr_i\}$ is a interpreted as a normalised probability distribution, and any such mixing of causal processes will again be causal.
\end{remark}

\seanignore{
\begin{remark}\label{ex:non-causal}
Semi-additive structure most often requires non-causal processes.  Moreover, probabilistic weights themselves are non-causal, except for the unique causal scalar $1$. `Sub-causal' processes (see Section~\ref{sec:subcausal}) can be interpreted as those which do not happen with certainty, but occur as one potential outcome with some probability, while more general processes lack such an interpretation, but are added for mathematical convenience, namely, such that $+$ is a total function.
\end{remark}
}
We now consider the role that these two structures play in quantum theory.


\begin{example}[Quantum events]\label{ex:QuantMixed}
These can be modelled  as a category $\QuantMixed$ which has the same objects and monoidal product as $\QuantPure$. Note that the bounded operators on a Hilbert space $\hilbH$ define an ordered real vector space $\mathcal{B}(\hilbH)$  with associated positive cone $\mathcal{B}(\hilbH)^+$. Morphisms between these are then defined as linear order-preserving maps $f \colon \mathcal{B}(\hilbH)\to\mathcal{B}(\hilbH')$ which moreover are completely positive, meaning $f \otimes \id{\hilbK}$ preserves positivity of elements, for all objects $\hilbK$.
\seanignore{
\[
f\otimes \id{\hilbK} (\mathcal{B}(\hilbH\otimes \hilbK)^+) \subseteq \mathcal{B}(\hilbH'\otimes \hilbK)^+
\]
for all objects $\hilbK$. }
 The semi-additive structure is defined as the usual addition of linear maps.
States therefore correspond to $\rho \in \mathcal{B}(\hilbH)^+$ (i.e. are density matrices), effects can be written in terms of states via the trace inner product
$\langle\rho, \_\rangle=\text{tr}(\rho\_)$
(i.e.~are POVM elements) and scalars are positive real numbers. The discarding effect is given by $\text{tr}(\mathbb{I}\_)$ and so causal states have trace $1$ and general causal morphisms are trace preserving (i.e. are CPTP maps).
This category can also be defined in terms of Selinger's $\CPM$ construction~\cite{SelingerCPM}, which we will return to in Section~\ref{sec:Daggers}.
There is an embedding of $\QuantPure$ in $\QuantMixed$ sending a process $f \colon \hilbH \to \hilbH'$ to the map $f\circ - \circ f^\dagger:\mathcal{B}(\hilbH)\to \mathcal{B}(\hilbH')$.
\seanignore{
Processes $f\in\QuantPure(\hilbH, \hilbH')$ can be embedded in $\QuantMixed$ as:
\[
f\circ - \circ f^\dagger:\mathcal{B}(\hilbH)\to \mathcal{B}(\hilbH')
\]
}
\end{example}

\seanignore{
\begin{remark}
In addition to the case already made in Remark \ref{ex:non-causal} for considering non-causal processes, note also that the quantum events of Example \ref{ex:QuantMixed} typically won't be causal.  The causal ones are the trace-preserving maps, and hence, for example, there are no causal effects  \cite{CKpaperI, CKBook} aside from the discarding effect itself.  However, several non-causal events bunched together as a probabilistic process can be causal e.g.~quantum measurements.  It is therefore very useful to have non-causal processes around, since one does want to be able to talk about
`what just happened'.
\end{remark}
}

\begin{remark}
The relationship between $\QuantPure$ and $\QuantMixed$ as described in the above example can be viewed in two different ways. Firstly note that $\QuantPure$ is a subcategory of $\QuantMixed$ and then it is a standard result in quantum information that there are two equivalent ways to write any general quantum transformation $f \colon A\to B$ in terms of processes in $\QuantPure$, firstly, via the Kraus decomposition, as a sum of pure transformations, \johnignore{$\{a_k\colon A \to B\}$, i.e.,} $f=\sum_k a_k$, and secondly, via Stinespring dilation, as a pure process \johnignore{$g:A \to B\otimes C$} with an extra output \johnignore{$C$} which is discarded,
$f=\rho_B \circ (\id{B}\otimes \discard{C})\circ g$ where $\rho_B$ is the monoidal coherence isomorphism.
Therefore all of the processes of $\QuantMixed$ can be obtained from those in $\QuantPure$ either by means of the semi-additive or discarding structure.

This is an important feature of quantum theory, for example, \johnignore{the fact that any mixed state can be represented as a pure bipartite state where one part is discarded is known as the} in the form of the \emph{purification} postulate \cite{Chiri1} which has been used as an axiom in reconstructing quantum theory \cite{Chiri2}. More generally, the $\CPM$ construction provides a recipe for producing categories with discarding and a form of purification, see Section~\ref{sec:Daggers}. Conceptually, this provides both an `internal' and `external' view on the origins of general quantum transformations, which this paper develops with the two constructions.
\end{remark}

Here are some other example theories which serve as useful points of comparison for quantum theory.

\begin{example}[Probabilistic classical events]\label{ex:class_prob_theory}
These can be modelled in the category $\ClassProb$, objects are natural numbers $n$, and morphisms $f\colon n\to m$ are $n\times m$ matrices with positive real entries. The monoidal unit is $1$ and the monoidal product is $n\otimes m = nm$. Semi-additive structure is provided by the matrix sum. States are therefore column vectors with positive real elements, effects are row vectors and scalars are positive real numbers.  The discarding effects are those of the form $(1,1,...,1)$ such that causal states are normalised probability distributions over a finite set and general causal morphisms are stochastic matrices.
\end{example}

\begin{example}[Possibilistic classical events] \label{example:Rel}
These can be modelled in the symmetric monoidal category $\Rel$ of sets and relations. Here objects are sets $A, B, \dots$,  morphisms $R \colon A \to B$ are relations from $A$ to $B$, i.e.~subsets $R \subseteq A \times B$, sequential composition of relations $R \colon A \to B$ and $S \colon B \to C$ is defined by:
$
S \circ R = \{(a,c) \in A \times C \mid R(a,b) \wedge S(b,c) \}
$
and parallel composition is the set-theoretic product $A \otimes B = A \times B$, so that the monoidal unit is given by the singleton set $I = \{\star\}$. In particular, the scalars $\sclr \colon I \to I$ are the Booleans $\{0,1\}$. $\Rel$ has discarding, with $\discard{A} \colon A \to I$ given by the unique relation with $a \mapsto \star$ for all $a \in A$.
Hence, causal relations $R \colon A \to B$ are those satisfying $\forall a \ \exists b \ R(a,b)$.
This theory is semi-additive under the union of relations $R + S := R \vee S$.
\end{example}

\begin{example}[Modal quantum events]
The events in modal quantum theories \cite{schumacher2012modal,schumacher2016almost} can be modelled in the symmetric monoidal category where the objects are lattices of subspaces of finite dimensional vector spaces over a particular finite field $\mathbb{Z}_p$, where the choice of field defines a particular modal theory $\Modalp$. Morphisms are $\lor$ and $\perp$ preserving maps between these lattices. The monoidal product is inherited from the tensor product of the underlying vector spaces. The monoidal unit is the lattice of subspaces of a 1D vector space which gives two scalars $0$ and $1$ interpreted as impossible and possible respectively. This allows us to define zero-morphisms by $0_{A,B}(a) = 0$ for all $a$.
The discarding effect $\discard{A} \colon A\to I$ can then be defined by:
$\discard{A}\circ a=0 \iff a = 0_{IA}$
where $a \colon I\to A$. This is again a semi-additive category with $(f\lor g)(a) = f(a) \vee g(a)$ for $a \in A$.
In fact, $\Rel$ can be viewed as $\Modalp$ for the case `$p = 1$' in a certain precise sense~\cite{cohn2004projective}.
\seanignore{
define by:
\[
(f\lor g)\circ a := (f\circ a)\lor (g\circ a)
\]
for all $a \colon I\to A$.
}
\end{example}


\begin{example}
Any semi-ring $(R, +, 0, 1)$ forms a one object, semi-additive discard category, with $\discard{} = 1$.
\end{example}

\section{The Two Roads} \label{sec:tworoads}

We now introduce the two constructions which adjoin classicality to a theory. The first follows an external perspective, describing how one may build mixtures of the existing objects\johnignore{ of the theory}. The second follows the internal perspective, showing how classicality can emerge due to restrictions arising from decoherence.\johnignore{ We then derive when these two constructions coincide.}

\paragraph{The biproduct completion}

The first of these approaches is captured by a standard notion from category theory. Recall that, in any semi-additive category, a \emph{biproduct} of a finite collection $\{A_i\}^n_{i=1}$ of objects consists of an object $A= \bigoplus^n_{i=1} A_i$ and morphisms
$
\begin{tikzcd}
A_i \rar{\coproj_i} & \bigoplus^n_{i=1} A_i \rar{\pi_j} & A_j
\end{tikzcd}
$
satisfying:
\begin{eqnarray}
\label{eq:biprod_first}
\pi_i \circ \coproj_j = 0 \text{ for } i \neq j  \qquad & &\qquad
\pi_i \circ \coproj_i = \id{A_i}\\
\label{eq:biprod_sum}
 \sum^n_{i=1} \pi_i \circ \coproj_i &=&\id{A}
\end{eqnarray}

This makes $A$ both a product and coproduct of the objects $\{A_i\}^n_{i=1}$. An empty biproduct is the same as a \emph{zero object} --- an object $0$ which is both initial and terminal.


\begin{definition}
In a category with discarding a biproduct is \emph{causal} when its morphisms $\coproj_i$ are causal.
\end{definition}

In contrast, the structural morphisms $\pi_i$ will usually not be causal. Note that a semi-additive category has finite (causal) biproducts whenever it has a zero object and (causal) biproducts of pairs of objects.

\begin{examples}
$\ClassProb$ and $\Rel$ each have causal biproducts, given by the direct sum of vector spaces and disjoint union of sets, respectively. Any \em grounded biproduct category \em in the sense of Cho, Jacobs and Westerbaan ($\times2$) \cite{cho2015introduction} is a semi-additive category with discarding and causal biproducts.
\end{examples}

\johnignore{We may now describe the first universal construction of our study.}
\begin{definition}[Biproduct completion]
Any semi-additive category $\catC$ may be embedded universally into one with biproducts $\catC^{\oplus}$, its \emph{free biproduct completion} (see \cite[Exercise VIII.2.6]{MacLane}). The objects of $\catC^{\oplus}$ are finite lists $\langle A_1, \dots, A_n \rangle$ of objects from $\catC$, with the empty list forming a zero object. Morphisms $M \colon \langle A_1, \dots, A_n \rangle \to  \langle B_1, \dots, B_m \rangle$ are matrices of morphisms $\langle M_{i,j} \colon A_i \to B_j\rangle$ and composition is the usual one of matrices. The identity on $\langle A_1, \dots, A_n \rangle$ is the matrix with identities on its diagonal entries and zeroes elsewhere.
\end{definition}

There is a canonical full and faithful embedding of semi-additive categories $\catC \hookrightarrow \catC^{\oplus}$ given by $A \mapsto \langle A \rangle$. Moreover, this is universal in that any semi-additive functor $F \colon \catC \to \catE$ from $\catC$ to a semi-additive category $\catE$ with biproducts lifts to one $\hat{F} \colon \catC^{\oplus} \to \catE$, unique up to natural isomorphism.

\begin{example}
Generalising Example~\ref{ex:class_prob_theory}, the biproduct completion of a semi-ring $R$, seen as a one-object category, is the category $\MatR$ of $R$-valued matrices. In particular $\ClassProb = \Mat_{\mathbb{R^+}}$, while the category $\Mat_\boolB$ of Boolean valued matrices is equivalent to $\FRel$, the full subcategory of $\Rel$ on finite sets.
\end{example}

When $\catC$ is a category with discarding, $\catC^{\oplus}$ is also with
 $I = \langle I \rangle$ and $\discard{ \langle A_1, \dots, A_n \rangle} = \langle \discard{A_i} \colon A_i \to I \rangle^n_{i=1}$. To state its universal property, we will need the following notions.

\begin{definition}
A functor $F \colon (\catC, \discard{}) \to (\catD, \discard{})$ between categories with discarding is \emph{causal} when the morphism $\discard{F(I)} \colon F(I) \to I$ is an isomorphism and $F(\discard{A}) = \discard{F(I)} \circ F(\discard{A})$ for all objects $A$. A natural transformation is \emph{causal} when all of its components are.
\end{definition}

In particular, the embedding $\catC \hookrightarrow \catC^{\oplus}$ is causal, with the same universal property as before when restricted to causal biproducts, functors and natural transformations.

\paragraph{Splitting idempotents}

We now turn to decoherence, capturing it with the following categorical concept. Recall that, in any category, an \emph{idempotent} on an object $A$ is a morphism $p \colon A \to A$ satisfying $p \circ p =  p$. An idempotent \emph{splits} when it decomposes as $p = m \circ e$ for a pair of morphisms $e \colon A \to B$, $m \colon B \to A$, with $e \circ m = \id{B}$.
We then denote the splitting pair by $(m,e)$. Conversely, for any such pair of morphisms, $m \circ e$ is always an idempotent on $A$.

\begin{definition}
In a category with discarding, an idempotent splits \emph{causally} when it has a splitting $(m,e)$ with $m$ causal.
\end{definition}

\johnignore{Splittings for idempotents may always be added freely, using our second construction of interest.}

\begin{definition}[Karoubi Envelope]
Splittings for idempotents can always be added freely for any category $\catC$. The \emph{Karoubi Envelope} $\Split(\catC)$ of $\catC$ is the following category~(\cite{borceux1986cauchy}, see also~\cite{SelingerIdempotent}).
Objects are pairs $(A,p)$ where $p \colon A \to A$ is an idempotent. Morphisms $f \colon (A,p) \to (B,q)$ are morphisms $f \colon A \to B$ in $\catC$ satisfying
$f = q \circ f \circ p$
In particular, the identity on $(A,p)$ is $p$.
\end{definition}
$\Split(\catC)$ has discarding whenever $\catC$ does, given by $\discard{(A,p)} = \discard{} \circ p \colon (A,p) \to (I,\id{})$.
 In this case\seanignore{
 the causal idempotents are the most physically interesting, and
 } we write $\Splitcausal(\catC)$ for the full subcategory on objects $(A,p)$ with $p$ causal. If $\catC$ is semi-additive then so is $\Splitorcausal(\catC)$, with addition lifted from $\catC$. The key property of the construction is as follows: there is a full (causal) embedding $\catC \hookrightarrow \Splitorcausal(\catC)$ given by $A \mapsto (A,\id{})$, which gives every (causal) idempotent in $\catC$ a (causal) splitting in $\Splitorcausal(\catC)$. As before, $\Splitorcausal(\catC)$ is universal with this property.

\begin{remark}
We may interpret the $\Split(\catC)$ construction as introducing new objects by imposing a fundamental restriction on the allowed morphisms, (i.e. those satisfying $f= q \circ f \circ p$) in Section~\ref{sec:leaks} we discuss how such a restriction can arise due to `leaking' information.
\end{remark}

\paragraph{Comparing the constructions}

\seanignore{
}
It is now natural to ask what the relationship is between the $\catC^{\oplus}$ and $\Splitcausal(\catC)$ constructions, and when they coincide. We now answer this question for both $\Splitcausal(\catC)$ and the more generally definable $\Split(\catC)$ simultaneously. We will require the following weakening of the notion of a biproduct:

\begin{definition} \label{def:disjointembed_finitedecomp}
Let $\catC$ be a category with zero morphisms.
A \emph{disjoint embedding} of a finite collection of objects $\{A_i\}^n_{i=1}$ in $\catC$ is given by an object $A$ and morphisms $\coproj_i \colon A_i \to A$, $\pi_j \colon A \to A_j$ satisfying the first collection of biproduct equations~\eqref{eq:biprod_first}.
When $\catC$ is semi-additive, the morphism $p = \sum^n_{i=1} \coproj_i \circ \pi_i$ is then an idempotent on $A$. When $\catC$ also has discarding, a \emph{causal} disjoint embedding is one for which the morphisms $p$ and $\coproj_i$ are all causal.
\end{definition}

In particular, an empty (causal) disjoint embedding is just an object of the form $(A,0)$ in $\Splitorcausal(\catC)$. Note that a disjoint embedding is not necessarily a biproduct in $\catC$, since it may fail to satisfy~\eqref{eq:biprod_sum}.

\begin{example}
Any (causal) biproduct is in particular a (causal) disjoint embedding. Hence they are present in our examples $\ClassProb$, $\Rel$ and $\MatR$.
\end{example}

\begin{example} \label{example:disjoint_embed_in_quantum}
A motivating example is $\QuantMixed$, which has disjoint embeddings which are not biproducts. Given a collection of (finite-dimensional) Hilbert spaces $\{\hilbH_i\}_i$, we may form their Hilbert space direct sum $\bigoplus^n_{i=1} \hilbH_i$, which is their biproduct in the category $\FHilb$ of (finite-dimensional) Hilbert spaces and continuous linear maps. However this is no longer a biproduct of the $\{\hilbH_i\}_i$ in $\QuantMixed$, where linear maps from $\FHilb$ are identified up to global phase, but only a disjoint embedding. Physically, the distinction is that the addition in $\FHilb$ is given by superposition, while that of $\QuantMixed$ refers to mixing. This example generalises to categories of the form $\CPM(\catC)$, see Proposition~\ref{prop:disj_embed_and_CPM} later.
\end{example}

\begin{remark} \label{rem:dis_embed}
Disjoint embeddings can be understood as a property of our theory allowing for the encoding of classical data. Concretely, they provide the ability to store any collection of systems in a disjoint way in some larger system. In particular, by forming a disjoint embedding $C$ of $n$ copies of $I$ we may store an $n$-level classical system in $C$. However, the choice of $C$ is non-canonical: there can be many which need not be isomorphic.
\end{remark}

Abstractly, the significance of disjoint embeddings is the following.

\begin{theorem} \label{thm:SplitHasBiprod}
Let $\catC$ be semi-additive (with discarding). Then
$\Splitorcausal(\catC)$ has finite (causal) biproducts iff $\catC$ has (causal) disjoint embeddings.
\end{theorem}

\begin{proof}
We have already seen that $\Splitorcausal(\catC)$ is semi-additive, so the statement makes sense. Expanding the definitions shows that a (causal) disjoint embedding is precisely a (causal) biproduct of the form $(A,p) = \bigoplus^n_{i=1} (A_i, \id{})$ in $\Splitorcausal(\catC)$. Hence the conditions are clearly necessary. It's easy to see that an empty (causal) disjoint embedding $(A,0)$ is a zero object in $\Splitorcausal(\catC)$, so it suffices to check this category has binary (causal) biproducts.

Now for any pair of objects $(A_1, p_1)$, $(A_2, p_2)$ in $\Splitorcausal(\catC)$, by assumption the objects $A_1$ $A_2$ have a (causal) disjoint embedding $(A,p)$, with morphisms $\coproj_i \colon A_i \to A$ and $\pi_j \colon A \to A_j$.
Then $q = \sum^2_{i=1} \coproj_i \circ p_i \circ \pi_i$ is a (causal) idempotent on $A$.
Further, it is easy to check that $\coproj_i \circ p_i \colon (A_i, p_i) \to (A,q)$ and $p_i \circ \pi_i \colon (A, q) \to (A_i,p_i)$ are well-defined morphisms in $\Splitorcausal(\catC)$ making $(A, q)$ a (causal) biproduct $(A_1,p_1) \oplus (A_2, p_2)$.
\end{proof}

Observe that the idempotents arising from disjoint embeddings are those with the following property.

\begin{definition} \label{def:finite_decomp}
An idempotent $p \colon A \to A$ has a \emph{finite decomposition} when it may be written as a finite sum $p = \sum^n_{i=1} m_i \circ e_i$
\seanignore{
\[
p = \sum^n_{i=1} m_i \circ e_i
\]
}
of split idempotents $p_i = m_i \circ e_i$ for which $p_i \circ p_j = 0$ for $i \neq j$. Such a decomposition is \emph{causal} when all of the morphisms $m_i$ are causal.
\end{definition}

We can now determine when our two constructions coincide (cf.~\cite[Corollary 4.7]{SelingerIdempotent}).

%
\begin{corollary} \label{cor:IdemBiprodEquiv}
When $\catC$ has (causal) disjoint embeddings, there is a canonical full, semi-additive (causal) embedding $F \colon \catC^{\oplus} \to \Splitorcausal(\catC)$. Further, the following are equivalent:
\begin{enumerate}[itemsep=-1mm]
\item
\label{cond:Feq}
$F$ is a (causal) equivalence of categories;
\item \label{cond:caussplpit}

every (causal) idempotent splits (causally) in $\catC^{\oplus}$;
\item \label{cond:fin_decomp}
every (causal) idempotent in $\catC$ has a (causal) finite decomposition.
\end{enumerate}
When the causal form of these hold we say $\catC$ has the \emph{finite decomposition property}.
\end{corollary}

\begin{proof}
To see that $F$ exists, apply the universal property of $\catC^{\oplus}$ to the (causal) embedding $\catC \to \Splitorcausal(\catC)$, using Theorem~\ref{thm:SplitHasBiprod}. Concretely, $F$ acts on objects by $(A_1, \cdots, A_n) \mapsto \bigoplus^n_{i=1} (A_i, \id{})$.

\eqref{cond:Feq}$\Rightarrow$ \eqref{cond:caussplpit}: Since $F$ is a (causal) equivalence and (causal) idempotents split (causally) in $\Splitorcausal(\catC)$, they do in $\catC^{\oplus}$.
\eqref{cond:caussplpit} $\Rightarrow$ \eqref{cond:fin_decomp}:
It's easy to see that a (causal) idempotent $p \colon A \to A$ splits (causally) over $(A_i)^n_{i=1}$ in $\catC^{\oplus}$ iff it has a (causal) finite decomposition $p = \sum^n_{i=1} p_i$ with each $p_i$ splitting over $A_i$.

\eqref{cond:fin_decomp} $\iff$ \eqref{cond:Feq}: By construction, $F$ is (causally) essentially surjective on objects iff every object $(A,p)$ forms a (causal) biproduct $\bigoplus^n_{i=1} (A_i, \id{A_i})$ in $\Splitorcausal(\catC)$. By definition, this holds iff every $p$ has a (causal) finite decomposition.
\end{proof}

\begin{example}
When $\catC$ already has biproducts the embedding $\catC \hookrightarrow \catC^{\oplus}$ is an equivalence, and Corollary~\ref{cor:IdemBiprodEquiv} amounts to the fact that any finitely decomposable idempotent $p = p_1 + \dots + p_n$ in $\catC$ already splits over $\bigoplus^n_{i=1} A_i$, where $p_i$ splits over $A_i$, for each $i$.
\end{example}


\begin{remark}[Monoidal Structure] \label{rem:MonoidalStructure}
These results are compatible with monoidal structure whenever it is present, in a straightforward way:
\begin{itemize}[leftmargin=*,itemsep=-1mm]
\item\sloppypar
a \emph{(symmetric) monoidal category with discarding} $(\catC, \otimes, \discard{})$ is a (symmetric) monoidal category $(\catC, \otimes ,I)$ which is also a category with discarding $(\catC,\discard{}, I)$ with $I$ being the monoidal unit, for which all coherence isomorphisms are causal and
$\discard{A \otimes B} = \lambda_I \circ (\discard{A} \otimes \discard{B})$
for all objects $A, B$,
where $\lambda_I \colon I \otimes I \to I$ is the coherence isomorphism.
\item
a \emph{causal} (symmetric) monoidal functor $F \colon (\catC, \otimes, \discard{}) \to (\catD, \otimes, \discard{})$ is a causal functor which is strong (symmetric) monoidal, with causal structure isomorphisms $I \to F(I)$ and $F(A) \otimes F(B) \to F(A \otimes B)$.
\item
A \emph{semi-additive monoidal category} is a monoidal category which is monoidally enriched in $\CMon$. Explicitly, it is semi-additive with $f \otimes (g + h) = f \otimes g + f \otimes h$, $(f + g) \otimes h = f \otimes h + g \otimes h$ and $f \otimes 0 = 0 = 0 \otimes g$ for all morphisms $f, g, h$.
\seanignore{
\begin{align*}
f \otimes (g + h) &= f \otimes g + f \otimes h \\
(f + g) \otimes h &= f \otimes h + g \otimes h \\
f \otimes 0 &= 0 = 0 \otimes g
\end{align*}
}
for all morphisms $f, g, h$.
\end{itemize}

When $\catC$ is a semi-additive (symmetric) monoidal category (with discarding) so are each of $\catC^{\oplus}$ and $\Splitorcausal(\catC)$, and they satisfy the same universal properties with respect to such categories and (causal, symmetric) monoidal semi-additive functors and (causal) monoidal natural transformations between them. In particular, the functors and equivalences of Corollary~\ref{cor:IdemBiprodEquiv} are now monoidal ones.
\end{remark}

\paragraph{Results on Idempotent Splittings}

Before turning to quantum theory, we briefly consider some abstract results of use later. 
\seanignore{
} Firstly, recall that a category with a distinguished object $I$ is \emph{well-pointed} when for all $f, g \colon A \to B$ with $f \circ a = g \circ a$ for all $a \colon I \to A$, we have $f = g$. For any morphism $f \colon A \to B$ we set
$\text{Im}(f) := \{f \circ a \mid a \colon I \to A \}$.

\begin{example}
All of the categories with discarding $\QuantMixed$, $\ClassProb$, $\Rel$, $\MatR$ are well-pointed over their usual object $I$. \emph{Real quantum theory} provides a physically interesting non-well-pointed category~\cite{hardy2012limited}.
\end{example}

\begin{lemma} \label{lem:split_lem}
 Let $p, q \colon A \to A$ be (causal) idempotents in a well-pointed category $(\catC,I)$, with $\text{Im}(p) = \text{Im}(q)$.
 Then $p$ splits (causally) iff $q$ does, and $p$ has a (causal) finite decomposition iff $q$ does.{}
\end{lemma}

\begin{proof}
For all $a \colon I \to A$ we have $q \circ a = p \circ b$ for some $b$ and so $p \circ q \circ a = p \circ p \circ b = p \circ b = q \circ a$. Hence by well-pointedness $p \circ q = q$, and dually $q \circ p = p$ also. This states precisely that $(A,p)$ and $(A,q)$ are (causally) isomorphic in $\Splitorcausal(\catC)$. Then $p$ has a (causal) finite decomposition iff it forms a (causal) biproduct $\bigoplus^n_{i=1} (A_i, \id{})$, iff $q$ does. Taking $i=1$ shows that $p$ splits (causally) precisely when $q$ does.
\end{proof}

Next we observe that for many theories, including quantum theory and our other probabilistic examples, splittings for causal idempotents in fact suffice to provide splittings of a broader class. \seanignore{ (c.f. Remark~\ref{ex:non-causal}).}

\begin{definition} \label{def:subcausal} In a semi-additive category with discarding, a morphism $f \colon A \to B$ is \emph{sub-causal} when there is some $x \colon A \to I$ with
$\left(\,\discard{} \circ f\,\right) + x = \discard{}$.
\end{definition}

\begin{examples}
Any causal process is in particular sub-causal. In $\QuantMixed$, a process is sub-causal when it is trace non-increasing. A sub-causal process in $\ClassProb$ is a sub-stochastic matrix, i.e.~one whose columns have sum bounded by $1$. In $\Rel$ and $\Modalp$ every process is sub-causal.
\end{examples}

\begin{proposition} \label{prop:subcausal_idem}
Let $\catC$ be a semi-additive category with discarding satisfying: 
\begin{itemize}[itemsep=-1mm]
\item
Cancellativity: $f + g = f + h \implies g = h $;
\item
For every non-zero $f \colon A \to B$ there exists $a \colon I \to A$ with $\discard{} \circ f \circ a = \id{I}$.
\end{itemize}
If causal idempotents causally split in $\catC$, so do (non-zero) sub-causal idempotents.
\end{proposition}

\begin{proof}
Let $p \colon A \to A$ be a non-zero sub-causal idempotent, with $(\discard{} \circ p) + x = \discard{}$. Note that we have
$\discard{} \circ p = (\discard{} \circ p + x) \circ p = (\discard{} \circ p) + (x \circ p)$
and so by cancellativity $x \circ p = 0$.
By assumption there exists $a \colon I \to A$ with $\discard{} \circ p \circ a = 1$. One may then check that $q = p + (p \circ a \circ x)$ is a causal idempotent satisfying $p \circ q = q$ and $q \circ p = p$. These ensure that if $q$ has splitting $(m, e)$ then $p$ has splitting $(m, e \circ p)$.
\end{proof}

\section{Quantum theory} \label{sec:quantum}

We now turn to the main result of this paper, that for the example of quantum theory our two constructions coincide, both leading to the (symmetric monoidal) category of finite dimensional C*-algebras.

\johnignore{
Recall that a \emph{C*-algebra} is a Banach algebra $(A,\|\_\|)$ over $\mathbb{C}$ with an involution $(-)^*$ compatible with complex conjugation, and which satisfies $\|A^*A\|=\|A^*\|\|A\|$ for all elements $A$. Specifically, we are concerned here with finite-dimensional C*-algebras, which are of most interest for quantum information theory. These form a symmetric monoidal category as we now describe.}

\begin{example}[C*-algebras]
In the category $\CStar$ objects are finite dimensional C*-algebras and morphisms completely positive linear maps (in the same sense as defined for $\QuantMixed$). The monoidal product is the standard tensor product of finite-dimensional C*-algebras and the monoidal unit is $\mathcal{B}(\mathbb{C})$. Semi-additive structure is provided by the standard sum of linear maps. Discarding effects are provided by the trace, and biproducts by the direct sum of C*-algebras.

Note that $\QuantMixed$ and $\ClassProb$ are each equivalent to full subcategories of $\CStar$, corresponding to C*-algebras of the form $\mathcal{B}(\hilbH)$, and the commutative C*-algebras, respectively. These subcategories can also be characterised in terms of leaks, see Section~\ref{sec:leaks}.
\end{example}

There is a well known classification result~\cite{bratteli1972inductive} stating that any finite dimensional C*-algebra is isomorphic to a direct sum of complex matrix algebras. It is therefore unsurprising that our first road, the biproduct completion, leads to this category: see~\cite[Example 3.4.]{EPTCS171.7} for the details.

\begin{example}\label{ex:CStarFromBiproducts}
There is a monoidal, causal equivalence of categories $\CStar \simeq \QuantMixed^{\oplus}$.
\seanignore{
$\CStar$ is the biproduct completion of $\QuantMixed$.
}
\end{example}

The second road, however, requires some more work.

\begin{proposition} \label{prop:CPTPIdemSplitFCStar}
$\QuantMixed$ has the finite decomposition property.
\end{proposition}
\begin{proof}
We saw in Example~\ref{example:disjoint_embed_in_quantum} that $\QuantMixed$ has causal disjoint embeddings.
Now let $p \in \QuantMixed(\hilbH,\hilbH)$ be a causal idempotent.
Then $p$ is an idempotent, trace-preserving, completely positive linear map on $\mathcal{B}(\hilbH)$.
 By \cite[Theorem 5]{blume2010information}, there is a decomposition $\hilbH \simeq \bigoplus_k A_k \otimes B_k$ and set of positive semidefinite matrices $\tau_K$ on $B_k$ with:
\[
p(\mathcal{B}(\hilbH)) = \{ \textstyle{\sum}_k M_{A_k} \otimes \tau_k \colon M_{A_k} \in \mathcal{B}(A_k) \}
\]
Without loss of generality, assume $\tau_k \neq 0$ and so $\text{Tr}(\tau_k) \neq 0$, for all $k$.
Then $p(\mathcal{B}(\hilbH)) = q(\mathcal{B}(\hilbH))$ for the causally finitely decomposable idempotent $q = \sum_{k \in K} m_k \circ e_k$, where $m_k \colon \mathcal{B}(A_k) \to \mathcal{B}(\hilbH)$ is given by $M \mapsto ( 1 / \text{Tr}(\tau_k)) M \otimes \tau_k $, and $e_k \colon \mathcal{B}(\hilbH) \to \mathcal{B}(A_k)$ is given by $M \mapsto \text{Tr}_{B_k} ( \text{Tr}_{A_j \otimes B_j, j \neq k}( M))$.

Since states $\rho \in \QuantMixed(\mathbb{C},\hilbH)$
may be viewed as elements of $\mathcal{B}(\hilbH)$, on which $p$ is an idempotent map, this gives that $\text{Im}(p) = \text{Im}(q)$ in the well-pointed category $(\QuantMixed, \mathbb{C})$. Hence by Lemma~\ref{lem:split_lem}, since $q$ has a causal finite decomposition, so does $p$.
\end{proof}

By combining Example~\ref{ex:CStarFromBiproducts}, Proposition~\ref{prop:CPTPIdemSplitFCStar} and Corollary~\ref{cor:IdemBiprodEquiv}, we reach our main result:
\begin{corollary} \label{cor:quantum_main_result}
There is a monoidal, causal equivalence:
\[
\Splitcausal(\QuantMixed) \simeq \CStar \simeq \QuantMixed^\oplus
\]
\end{corollary}

\johnignore{\begin{proof}
Combine Example~\ref{ex:CStarFromBiproducts}, Proposition~\ref{prop:CPTPIdemSplitFCStar} and Corollary~\ref{cor:IdemBiprodEquiv}.
\end{proof}}

Hence all causal (trace-preserving) idempotents causally split in $\CStar$. By Proposition~\ref{prop:subcausal_idem}, the same in fact holds for all idempotents which are sub-causal, \textit{i.e.} trace-non increasing. However, it remains an open question whether all idempotents split.
\seanignore{
}

\section{Further Examples}

\paragraph{Classical probability theory}
\johnignore {Another example in which both constructions coincide is provided by classical probability theory, since in fact}
There are (monoidal, causal) equivalences:
\[
\ClassProb^\oplus \simeq \ClassProb \simeq \Split(\ClassProb) \simeq \Splitcausal(\ClassProb)
\]
The left hand equivalence holds since
\johnignore{, as we have seen,}
 $\ClassProb$ already has biproducts. Conversely, it follow from a Theorem of Flor (\cite[Theorem 2]{flor1969groups}, see also~\cite[Theorem 4]{demarr1974nonnegative}) that any idempotent $p$ in $\ClassProb$ has a finite decomposition $p = z_1 + \dots + z_k$ where the $z_i$ satisfy $z_i \circ z_j = \delta_{i,j} z_i$ and each are of rank one, hence splitting over $I$. In particular, \emph{every} idempotent in $\ClassProb$ has a finite decomposition and so splits, yielding the other equivalences. Hence, as in the quantum case, these constructions coincide.

\paragraph{Possibilistic theories}
In contrast, the constructions will generally fail to coincide in theories of a possibilistic nature, such as $\Modalp$ or $\Rel$.
\seanignore{
To make this precise, say that a theory is \emph{positive} when $\discard{} \circ (f + g) = 0 \implies f = g = 0$, for all morphisms $f, g$.
} By a \emph{possibilistic} theory we mean one in which the addition $+$ is idempotent, the scalars $\sclr \colon I \to I$ under $(\circ, +)$ are the Booleans $\{0,1\}$, and we have  $\discard{} \circ f = 0 \implies f=0$ for all morphisms $f$.
We will consider possibilistic theories with a particular physically motivated property.
Call a pair of states $\pds_0, \pds_1 \colon I \to A$ on $A$ \emph{perfectly distinguishable} when there exists a pair of effects $\pdeff{\pds_0}, \pdeff{\pds_1}$ on $A$ with $\pdeff{\pds_0} + \pdeff{\pds_1} = \discard{}$ and $\pdeff{\pds_i} \circ \pds_j = \delta_{i,j}$. We say that a theory satisfies \emph{perfect distinguishability} when every system not isomorphic to $I$ or a zero object $0$ has a pair of perfectly distinguishable states, and there exists at least one such system. Then we have the following (see App.~\ref{proof:pos_theory_counter} for proof):

\begin{proposition} \label{lem:pos_theory_counter}
Any possibilistic theory $(\catC, \discard{})$ with perfect distinguishability lacks the finite decomposition property.
\end{proposition}

\begin{example} \label{ex:Rel_Counter}
Both $\Modalp$ and $\Rel$ are possibilistic theories with discarding and perfect distinguishability, and so for these theories the constructions do not coincide.

For example, in $\Rel$, any set $A$ not isomorphic to $I = \{\star\}$ or $0 = \emptyset$ has at least two distinct elements, forming a pair of perfectly distinguishable states. Hence $\Rel$ lacks the finite decomposition property. Concretely, the causal idempotent~\eqref{eq:counterex_proj} is the relation on $\{0,1\}$ given by $0 \mapsto 0$, $1 \mapsto 0,1$, does not split.
\end{example}

\paragraph{Information units}
\johnignore{It is not just possibilistic models for which these two constructions differ,}
A somewhat contrived example can be motivated by considering the idea of an information unit for a theory:
\johnignore{An information unit is}
a particular object $U$ in a theory $\catC$ such that any process in the theory can be `simulated' on an $n$-fold monoidal product of $U$.
\johnignore{In quantum theory such an information unit is provided by the qubit.}
 The existence of an information unit has been used as a postulate in reconstructing quantum theory \cite{masanes2013existence} where $U$ is a qubit, and, moreover, underlies the circuit model of quantum computation. We define the information unit subtheory, $\catCU$, as a full subcategory
 \johnignore{of a theory with an information unit}
 restricting to objects of the form $U^{\otimes n}$.
 \johnignore{ and denote this subcategory as $\catCU$.}

\begin{example}
For the quantum case with $U=\mathbb{C}^2$ we find that
 $\QuantMixed_{\langle U \rangle}^\oplus\subsetneq{\CStar}\simeq\Splitcausal(\QuantMixed_{\langle U \rangle}).$
To see the equivalence note that we can obtain an arbitrary $n$-level quantum system by splitting a causal idempotent, i.e., consider $m$ such that $n\leq M := 2^m$ and consider a sub-causal projector onto an $n$ dimensional subspace then Proposition~\ref{prop:subcausal_idem} shows there is a causal idempotent which splits over the same system. On the other hand is not possible to construct (for example) a qutrit with direct sums or tensor products of qubits and hence the biproduct completion does not give the entirety of $\CStar$. In fact, this result can be seen as a demonstration that the qubit is indeed an information unit for quantum theory as this provides a way to simulate any other quantum system on some composite of qubits.
\end{example}

\section{Idempotents from leaks}\label{sec:leaks}

Idempotents naturally arise, in symmetric monoidal theories with discarding, from
information leakage. \johnignore{Throughout we work in a symmetric monoidal category with discarding $(\catC, \otimes, \discard{})$.} 

\begin{definition}
A \em leak \em on an object $A$ is a morphism $l \colon A \to A \otimes L$ which has discarding as a right counit, that is: $\rho_A \circ (\id{A} \otimes \discard{L}) \circ l = \id{A}$. It follows that all leaks are causal processes.
\seanignore{
\beq\label{eq:leal}
l \colon A\to A\otimes L
\eeq
which has discarding as a right counit, that is:
\[
\rho_A \circ (\id{A}\otimes \discard{L})\circ l =\id{A}
\]
}
\end{definition}

\begin{example} In $\QuantMixed$ all leaks are \emph{constant}: i.e.~of the form $(\id{A}\otimes \sigma) \circ \rho_A^{-1}$
\seanignore{
\[
(\id{A}\otimes \sigma) \circ \rho_A^{-1}
\]
}
for a causal state $\sigma$ on $L$. 
\end{example}
\begin{example}
A \emph{broadcasting map}~\cite{Nobroadcast, CKpaperI} on an object $A$ is a leak $l \colon A \to A \otimes A$, \textit{i.e.} with $L = A$, for which discarding is also a left counit. \johnignore{The function $X \to X \times X$ on a set $X$ with $x \mapsto (x,x)$ lifts to a broadcasting map in} Both of our classical theories $\ClassProb$ and $\Rel$ have such a map.
\end{example}
\begin{remark}
In fact, minimal and maximal leakage characterises quantum and classical theory respectively\johnignore{while maximal leakage characterises classical theory}  \cite{QCILeaks}. Specifically, $\QuantMixed$ and $\ClassProb$ are equivalent to the full subcategories of $\CStar$ on objects with only constant leaks, and on objects with maximal leaks, i.e.~broadcasting maps, respectively.
\end{remark}
It is natural, given any theory $(\catC,\discard{})$, and for each system $A$ a chosen process
$l_A \colon A \to A \otimes L_A$
to construct a new theory $\leakminC$ in which $l_A$ represents the ongoing `leakage' of that system into the environment. The result of this leakage on $A$ would be the process:
\begin{equation}
\label{eq:idemp-leakconstr}
\iota_A:= \rho_A \circ (\id{A}\otimes \discard{L_A})\circ l_A
\end{equation}
Processes $A \to B$ in the new theory should then consist of applying the leakage to the inputs and outputs of every process $f \colon A \to B$ in the original theory, i.e.\johnignore{ to be of the form} $\iota_B\circ f\circ \iota_A$. Demanding the $\iota_A$ be idempotent ensures that $\leakminC$ is a category ---
the full subcategory of $\Splitcausal(\catC)$ given by the objects $(A, \iota_A)$. It is again symmetric monoidal whenever our choice of leaks $l_A$ ensures that $\iota_{A\otimes B}= \iota_A\otimes \iota_B$ for all $A, B$. In this case, $\iota_A$ becomes the identity on $A$, guaranteeing that $l_A$ provides
a leak in the new theory,  $[(\iota_A\otimes \iota_{L_A})\circ l_A \circ \iota_A]$.
\johnignore{
\[
\rho_A \circ (\id{A}\otimes \discard{L_A})\circ[(\iota_A\otimes \iota_{L_A})\circ l_A \circ \iota_A] = \iota_A\circ \iota_A\circ\iota_A = \iota_A
\]
}

Alternately, starting again from $\catC$, if we want to obtain a theory which can describe all possible systems that could arise from some leakage then we must instead take systems to be all possible pairs $(A, l_A)$ with the above property, i.e.~such that $\iota_A$ is idempotent. But this is none other than $\Splitcausal(\catC)$. Indeed, \emph{any} causal idempotent $p \colon A \to A$ is of the form (\ref{eq:idemp-leakconstr}) for some $l_A$, for example by taking $l_A := \rho_A^{-1} \circ p \colon A \to A \otimes I$. A more insightful manner is by taking $l_A$ to be any `purification' of $\iota_A$, that is, any pure process satisfying~\eqref{eq:idemp-leakconstr}. The existence of such a process is guaranteed in quantum theory by the Stinespring dilation theorem, and more generally for theories arising from the $\CPM$-construction \cite{CPer}, or, those satisfying the purification postulate of \cite{Chiri1}.
Hence we have generalised the idea that decoherence results from information-leaking, from the specific case of quantum theory to a much broader class of theories. There is another important conclusion we can draw from idempotents resulting from leaks.

\begin{example}\label{ex:CP*1}
A result of Vicary \cite{VicaryCstar} states that finite dimensional C*-algebras are precisely \em (dagger special) Frobenius structures \em in $\FHilb$: comonoids
$\delta \colon A \to A \otimes A$
which additionally satisfy (together with their `daggers', see Section~\ref{sec:Daggers} below) two equations called the \emph{dagger Frobenius law}: and \emph{speciality}.

\seanignore{
the \emph{dagger Frobenius law}:
\[
(\id{A} \otimes \delta^\dagger)\circ \alpha \circ (\delta\otimes \id{A}) =(\delta^\dagger\otimes \id{A})\circ \alpha^{-1} \circ (\id{A}\otimes\delta)
\]
where $\alpha$ is the appropriate coherence isomorphism, and \emph{speciality}:
\[
\delta^\dagger \circ \delta = \id{A}
\]
}
\end{example}

However, we now know that the much weaker concept of a process
$l \colon A \to A \otimes L$
in the directly physically interpretable category $\QuantMixed$, for which $\iota$ is a causal idempotent as above, is already sufficient to guarantee the rest of this comonoid structure. Hence, one is tempted to deduce that the essential physical structure of C*-algebras is captured by that of a leak.

\section{Comparison with  dagger categorical approaches}\label{sec:Daggers}

The use of biproducts and idempotent splittings to model hybrid quantum-classical systems has in fact already been studied by Selinger~\cite{SelingerIdempotent},
and Heunen, Kissinger and Selinger~\cite{EPTCS171.7}. Crucially, however, these previous works have relied on a feature of quantum theory not present in general physical theories -- the existence of a \emph{dagger}.\seanignore{
Our work bares close similarities with other uses of biproducts and idempotent splittings for the modelling of quantum-classical systems.
In particular, both notions are studied in this context by Selinger in~\cite{SelingerIdempotent}, while Corollary~\ref{cor:quantum_main_result} is a strengthening of an existing result of Heunen, Kissinger and Selinger~\cite{EPTCS171.7}. Crucially, however, these previous approaches have relied on a particular feature of quantum theory not present in general physical theories, namely the existence of a \emph{dagger}.
}
Recall that a \emph{dagger category} $(\catC, \dagger)$ is a category $\catC$ coming with an involutive, identity on objects functor $\dagger \colon \catC^{\op} \to \catC$. A \emph{dagger compact} category is additionally symmetric monoidal, in a way compatible with the dagger, with every object having a \emph{dagger dual} -- see~\cite{AC1,SelingerCPM} for details.
\seanignore{
Then a dagger semi-additive category (dagger monoidal category etc.)~is one coming with a dagger respecting the additive (monoidal, etc.) structure in a straightforward sense. In particular, a \emph{dagger compact} category is a dagger symmetric monoidal category in which each object $A$ comes with a so-called \emph{dagger dual} object $A^*$ -- see~\cite{AC1,SelingerCPM} for details.
}

\begin{examples}
 $\Rel$, $\QuantPure$, $\QuantMixed$, $\FHilb$, $\ClassProb$ and $\CStar$ are all dagger compact categories. In $\Rel$ the dagger is given by relational converse, and in the other cases it extends the adjoint of linear maps.
 \end{examples}


The results of Section~\ref{sec:tworoads} can easily be adapted to include daggers, as follows. A \emph{dagger} idempotent is one $p$ with $p = p^{\dagger}$, a \emph{dagger} splitting $p = m \circ e$ is one with $e = m^{\dagger}$, and a \emph{dagger} disjoint embedding or biproduct is one with $\pi_i = \coproj_i^{\dagger}$  for all $i$. Then $\catD^{\oplus}$ and $\Splitdagcausal(\catD)$, the full subcategory of $\Splitcausal(\catD)$ on the causal dagger idempotents, satisfy the same universal properties as before with respect to dagger-respecting functors, and Corollary~\ref{cor:IdemBiprodEquiv} becomes an equivalence of dagger categories $\catD^{\oplus} \simeq \Splitdagcausal(\catD)$.
\seanignore{
The results of Section~\ref{sec:tworoads} do without daggers, they may be easily be adapted to include them, as follows.
In a dagger category, a \emph{dagger idempotent} $p \colon A \to A$ is an idempotent satisfying $p = p^{\dagger}$, and it has a \emph{dagger} splitting when it splits as $(m, e)$ with $m = e^{\dagger}$. We write $\Splitdagcausal(\catD)$ for the full subcategory of $\Splitcausal(\catD)$ determined by the causal dagger idempotents.
Similarly, a \emph{dagger} disjoint embedding or biproduct $(A, A_i,\coproj_i,\pi_j)$ is one for which $\pi_i = \coproj_i^{\dagger}$ holds, for all $i$.
Then $\catD^{\oplus}$ and $\Splitdagcausal(\catD)$ satisfy their usual universal properties with respect to functors and natural transformations respecting the dagger. Further,
the analogue of Corollary~\ref{cor:IdemBiprodEquiv} is now a dagger-respecting equivalence of dagger categories $\catD^{\oplus} \simeq \Splitdagcausal(\catD)$.
}
\paragraph{The $\CPM$ construction}\label{sec:CPM}

Given a dagger compact category $\catC$ of `pure' processes, we may construct a new one, $\CPM(\catC)$, with the same objects but interpreted as consisting of mixed processes~\cite{SelingerCPM}.
An axiomatization of the construction closely resembling our treatment is provided by the following notion~\cite{SelingerAxiom, CPer}.
An \emph{environment structure} $(\catD, \catD_{\pure}, \discard{})$ consists of a dagger compact category $\catD$ with discarding $\discard{}$ respecting the dagger compact structure, along with a chosen dagger compact subcategory $\catD_{\pure}$ satisfying an axiom relating $\discard{}$ and $\dagger$.
It satisfies \emph{purification} whenever every morphism in $\catD$ is of the form
$\lambda \circ (\discard{} \otimes \id{}) \circ f
$ for some morphism $f$ in $\catD_{\pure}$. In this case there is a (dagger monoidal) isomorphism of categories $
\catD \simeq \CPM(\catD_{\pure})$. Conversely, every $\catD = \CPM(\catC)$ arises in this way.

The key example is that $\QuantMixed$ and $\QuantPure$ form an environment structure with purification, with
$\QuantMixed \simeq \CPM(\QuantPure) \simeq \CPM(\FHilb)$.
\seanignore{
\begin{example}
$\QuantMixed$ and $\QuantPure$ form an environment structure with purification, and we have:
$\QuantMixed \simeq \CPM(\QuantPure) \simeq \CPM(\FHilb)$.
\end{example}
}
We have seen that, while $\FHilb$ has biproducts, $\QuantMixed$ merely has disjoint embeddings. In fact, these suffice to deduce properties of the $\CPM$ construction previously shown using biproducts (cf.~\cite[Theorem 4.5]{SelingerIdempotent}) - see App.~\ref{proof:pos_theory_counter} for a proof.

\begin{proposition} \label{prop:disj_embed_and_CPM}
If $\catC$ is a dagger compact category with zero morphisms and dagger (causal) disjoint embeddings, so is $\CPM(\catC)$.
If $\catC$ is also dagger semi-additive so is $\CPM(\catC)$, and then there is a full embedding
$\CPM(\catC)^{\oplus} \hookrightarrow \Splitdagcausal(\CPM(\catC))$.
\end{proposition}
\seanignore{
\begin{proof} 
Use that there is a dagger functor $\catC \to \CPM(\catC)$ which preserves zeroes and is surjective on objects.
When $\catC$ has disjoint embeddings, $\CPM(\catC)$ is closed under addition in $\catC$ -- this is proven just as for biproducts in~\cite[Lemma 4.7(c)]{SelingerCPM}. Finally, use (the dagger version of) Theorem~\ref{thm:SplitHasBiprod}.
\end{proof}
}
\seanignore{
The existence of disjoint embeddings is proven
In \cite[Lemma 4.7(c)]{SelingerCPM} it is established that when $\catD$ has dagger biproducts, $\CPM(\catD)$ is closed under addition in $\catD$. The proof for disjoint embeddings is the same. The final statement uses (the dagger version of) Theorem~\ref{thm:SplitHasBiprod}.
\end{proof}
}

\paragraph{The $\CPs$~construction}
There is another treatment of hybrid classical-quantum systems applicable to  any dagger compact category $\catC$.
Dagger special Frobenius structures in $\catC$ (see Example~\ref{ex:CP*1})  form the objects of a category $\CPs(\catC)$, with the motivating example being that
$\CPs(\FHilb) \simeq \CStar$~\cite{VicaryCstar}.
\seanignore{
result from \cite{VicaryCstar}, extending Example \ref{ex:CP*1}:
\[
\CPs(\FHilb) \simeq \CStar
\]
}

An axiomatisation of the $\CPs$ construction, extending that for $\CPM$, has been given by Cunningham and Heunen~\cite{cunningham2015axiomatizing}. Given an environment structure, observe that for each Frobenius structure $(A,\delta)$
in $\catD_{\pure}$, the morphism
$\lambda_A\circ (\discard{} \otimes \id{}) \circ \delta$
is a causal dagger idempotent. A \emph{decoherence structure} is a choice of dagger splitting for each such idempotent, in a way compatible with the monoidal structure. Immediately, we have the following.


\begin{proposition} \label{prop:decoherence_structure}
If $\catD$ has an environment structure, $\Splitdagcausal(\catD)$ has a canonical decoherence structure.
\end{proposition}

Under a mild extra assumption known as \emph{positive dimensionality}, any decoherence structure induces a faithful (causal, dagger) functor \cite{cunningham2015axiomatizing}:
$\CPs(\catD_{\pure}) \to \Splitdagcausal(\catD)$.
Moreover, results of Heunen, Kissinger and Selinger~\cite{EPTCS171.7} tell us that, whenever purification is satisfied, so that $\catD$ is of the form $\CPM(\catC)$ for some $\catC$, this is a full causal embedding, and that when $\catC$ has dagger biproducts the causal embeddings from each construction factor as:
\begin{equation} \label{eq:CPMCPsarrows}
\CPM(\catC)^{\oplus} \hookrightarrow \CPs(\catC) \hookrightarrow \Splitdagcausal(\CPM(\catC)) \hookrightarrow \Splitcausal(\CPM(\catC))
\end{equation}
Then using Corollary~\ref{cor:IdemBiprodEquiv} we have the following:

\begin{corollary}
Let $\catC$ have dagger biproducts, and suppose $\CPM(\catC)$ has the finite decomposition property. Then each of the inclusions~\eqref{eq:CPMCPsarrows} are (causal, monoidal) equivalences of categories.
\end{corollary}

In particular, the case $\catC=\FHilb$ gives equivalences:
\[
\QuantMixed^{\oplus} \simeq \CStar \simeq \Splitdagcausal(\QuantMixed) \simeq \Splitcausal(\QuantMixed)
\]
extending the result from~\cite{EPTCS171.7}.

\section{Conclusion} \label{sec:concl}

We explored two seemingly unrelated categorical representations of classicality, through the $\Splitcausal(\catC)$ and $\catC^{\oplus}$ constructions, finding clear properties ensuring that they coincide -- the existence of disjoint embeddings and the finite decomposition property.
In Corollary~\ref{cor:quantum_main_result} we showed that these hold for quantum theory, strengthening the result of Heunen-Kissinger-Selinger~\cite{EPTCS171.7} by removing all mention of daggers.
This strengthening allows for a clear physical interpretation. In particular, we saw that we may obtain all of the usual C*-algebraic structure simply from the concept of leaking information (Section~\ref{sec:leaks}).

The generality of the approach here leaves many interesting open questions. Firstly, note that the results of Section~\ref{sec:tworoads} do not crucially rely on the discarding structure.
Though the restriction to sub-causal idempotents is well motivated physically, as it is only these that have an interpretation as probabilistic outcomes, an obvious open question is whether it is needed.
Do \emph{all} idempotents in $\CStar$ split?
Also, since these results do not rely on dagger, compact closed or even monoidal structure, they are applicable to the infinite-dimensional case.
Do all (sub-causal) idempotents split (causally) in the category of arbitrary C*-algebras (or von Neumann algebras) and completely positive maps?

Finally, while the existence of disjoint embeddings has a reasonable physical interpretation (see Remark~\ref{rem:dis_embed}), our main result also relied on the finite decomposition property, which is less well understood.
Are there any clear physical principles which allow one to deduce the finite decomposition property?
A positive answer to this would yield a more insightful proof of Corollary~\ref{cor:quantum_main_result}. Conversely, what other interesting physical consequences does this property have?


\appendix


\section{Proof of Lemma \ref{lem:pos_theory_counter} }\label{proof:pos_theory_counter}

In this appendix we prove that any possibilistic theory $(\catC, \discard{})$ with perfect distinguishability lacks the finite decomposition property.

\proof
It is easy to see that the biproduct completion $\catC^{\oplus}$ is again a possibilistic theory satisfying perfect distinguishability, and hence it suffices to show that causal idempotents do not split in $\catC$.
Pick a system $A$ possessing a pair of states $\{\pda_0, \pda_1\}$ perfectly distinguishable by some effects $\{\pdeff{\pda_0}, \pdeff{\pda_1}\}$. Then since $\pdeff{\pda_1} + \discard{} = \discard{}$, the following defines a causal idempotent on $A$:
\begin{equation} \label{eq:counterex_proj}
p =  \pda_0 \circ \discard{} + \pda_1 \circ \pdeff{\pda_1}
\end{equation}
Suppose $p \colon A \to A$ has a splitting $(m,e)$ over an object $B$.
Next suppose $B$ has two perfectly distinguishable states $\pdb_0, \pdb_1$, via the effects $\{\pdeff{\pdb_0}, \pdeff{\pdb_1}\}$. We define new states and effects on $A$ by $\pdc_i = m \circ \pdb_i$ and $\pdeff{\pdc_i} = \pdeff{\pdb_i} \circ e$, respectively, for $i = 1, 2$.
Then:
\[
0 = \pdeff{\pdc_0} \circ \pdc_1 = \pdeff{\pdc_0} \circ p \circ \pdc_1 = (\pdeff{\pdc_0} \circ \pda_0)\circ (\discard{} \circ \pdc_1) + (\pdeff{\pdc_0} \circ \pda_1)\circ(\pdeff{\pda_1} \circ \pdc_1)
\]
Since the scalars are the Booleans, and the $b_i$ and hence $c_i$ are non-zero, in particular we must have $\pdeff{\pdc_0} \circ \pda_0 = 0$. Dually, $\pdeff{\pdc_1} \circ \pda_0  = 0$ holds also.
But then $\discard{} \circ e \circ \pda_0 = (\pdeff{\pdb_0} + \pdeff{\pdb_1}) \circ e \circ \pda_0 = (\pdeff{\pdc_0} + \pdeff{\pdc_1}) \circ \pda_0 = 0$. By positivity, $e \circ \pda_0 = 0$, and hence $\pda_0 = p \circ \pda_0 = m \circ e \circ \pda_0 = 0$, contradicting $\pdeff{\pda_0} \circ \pda_0 = 1$.

We conclude that no such pair of states exists on $B$. Since $p \neq 0$, $B$ cannot be a zero object. Hence we must have $B \simeq I$ and so $p = x \circ y$ for some state $x$ and effect $y$ on $A$, respectively. But then, for any effect $z$ on $A$, since the scalars are the Booleans we must have either $z \circ p = 0$ or $z \circ p = y$. Hence $\pdeff{\pda_1} = \pdeff{\pda_1} \circ p = y = \discard{} \circ p = \discard{}$, and so $\pdeff{\pda_1} \circ \pda_0 = 1$, a contradiction.
\endproof

\section{Proof of Proposition~\ref{prop:disj_embed_and_CPM}} \label{proof:CPM_disj_embed}

Suppose that $\catC$ is dagger compact with zero arrows and disjoint embeddings. There is always a dagger functor $\catC \to \CPM(\catC)$ which preserves zeroes and is surjective on objects, and hence $\CPM(\catC)$ is also. When $\catC$ has disjoint embeddings, $\CPM(\catC)$ is closed under addition in $\catC$ -- this is proven just as for biproducts in~\cite[Lemma 4.7(c)]{SelingerCPM}. Finally, (the dagger version of) Theorem~\ref{thm:SplitHasBiprod} gives an embedding $\CPM(\catC)^{\oplus} \hookrightarrow \Splitdagcausal(\CPM(\catC))$.
. \endproof
\bibliographystyle{eptcs}
\bibliography{main}

\end{document}